\documentclass[10pt,a4paper]{article}
\usepackage[margin=22mm]{geometry}
\usepackage[T1]{fontenc}
\usepackage{lmodern}
\usepackage{amsmath,amssymb,amsthm,booktabs,fancyvrb}
\usepackage{graphicx}
\usepackage{float}
\usepackage[labelfont=bf,skip=4pt]{caption}
\usepackage[authoryear,round]{natbib}
\usepackage[hidelinks]{hyperref}
\usepackage{microtype}
\hypersetup{
  pdftitle={Three Improved Lower Bounds for Binary Error-Correcting Codes},
  pdfauthor={Kazuki Ota, Takayuki Osa, Tatsuya Harada}
}
\theoremstyle{plain}
\newtheorem{definition}{Definition}[section]
\newtheorem{proposition}[definition]{Proposition}

\begin{document}
\begin{center}
{\LARGE\bfseries Three Improved Lower Bounds for\\
Binary Error-Correcting Codes\par}
\vspace{8pt}
{\normalsize Kazuki Ota\textsuperscript{1,2,*}\quad
Takayuki Osa\textsuperscript{2}\quad
Tatsuya Harada\textsuperscript{1,2}\par}
\vspace{3pt}
{\small\textsuperscript{1}The University of Tokyo, Japan\par}
{\small\textsuperscript{2}RIKEN Center for Advanced Intelligence Project, Japan\par}
\end{center}
\begingroup
\renewcommand{\thefootnote}{*}
\footnotetext[1]{Corresponding author: Kazuki Ota
(\href{mailto:ota@mi.t.u-tokyo.ac.jp}{\nolinkurl{ota@mi.t.u-tokyo.ac.jp}}).}
\endgroup
\vspace{3pt}

\begingroup
\setlength{\leftmargini}{0.10\textwidth}
\begin{abstract}
Error-correcting codes are a foundational technology for representing
information robustly against noise in communication and data storage.
The binary code packing problem, which seeks to maximize the number of
codewords at a fixed code length and minimum distance, is a fundamental
problem in coding theory that remains open for many parameter choices.
We give explicit binary codes with parameters
$(n,M,d)=(22,84,9)$, $(24,196,9)$, and $(25,65,11)$, where $n$, $M$, and $d$
denote the code length, number of codewords, and minimum distance,
respectively. These establish the lower bounds
$A_2(22,9)\ge84$, $A_2(24,9)\ge196$, and $A_2(25,11)\ge65$,
where $A_2(n,d)$ denotes the maximum possible number of codewords, improving
the respective lower bounds 80, 192, and 64 in the public table of general
binary codes. Most notably, the third result improves the lower bound
supplied by a lexicographic construction for the first time in 66 years,
while the first two improve bounds that have stood for 21 years.
The new codes represent more messages without changing the code length or
error-correction guarantees, with gains compounding multiplicatively across
blocks. We release all codewords and an exhaustive verifier to enable
independent verification.
\end{abstract}
\endgroup

\section{Introduction}

Error-correcting codes are a foundational technology for recovering original
data from communication noise, media damage, and recognition errors. They
are widely used in applications such as digital communication, QR-code and
barcode recognition, and information storage. Since error-correcting codes directly affect the
reliability and efficiency of communication and storage systems, even small
improvements in code construction can have practical implications at scale.

The task of representing as much information as possible with short bit
strings while preserving error-correction capability can be formulated
mathematically as the code packing problem, which seeks a maximum-cardinality
set of length-$n$ binary words with pairwise Hamming distances of at least $d$.
Decades of coding-theory research have produced strong
constructions and theoretical upper bounds for many parameters
\citep{best,agrell,gijswijt}. Nevertheless, for many parameters, a gap remains
between the best-known constructions and the theoretical upper bounds.

In this work, we improve three lower bounds on the maximum size of binary
codes. Specifically, we give explicit codes containing 84, 196, and 65
codewords for $(n,d)=(22,9)$, $(24,9)$, and $(25,11)$, respectively,
surpassing the previous lower bounds of 80 \citep{ostergard},
192 \citep{elssel}, and 64 \citep{levenshtein}. These results improve two
lower bounds that have stood for 21 years and surpass the 64-codeword bound
supplied by a lexicographic construction dating to 1960, 66 years after its
introduction. With block length and minimum distance unchanged, the ratio of
message-sequence counts between the new and previous constructions grows
exponentially with the number of blocks. We provide all codewords and a
short exhaustive verifier, presenting these bounds in a form that permits
independent verification.

\section{Preliminaries}

The Hamming distance $d_H(x,y)$ is the number of coordinates at which
$x$ and $y$ differ. A binary $(n,M,d)$ code is a set
$S\subseteq\{0,1\}^n$ of $M$ distinct words whose pairwise distances
are at least $d$. Codes need not be linear.
This separation allows nearest-neighbor decoding to correct any pattern of
up to $\lfloor(d-1)/2\rfloor$ bit errors. For example, the four length-$5$
words in Figure~\ref{fig:error-correction-example} have minimum distance $3$,
so any single-bit error can be corrected.

For fixed code length $n$ and minimum distance $d$, increasing the number of
codewords allows more distinct messages to be represented without weakening
the error-correction guarantee. The binary code packing problem, formalized
in Definition~\ref{def:binary-code-packing}, seeks the largest possible set
of codewords satisfying these constraints.

\clearpage
\begin{figure}[H]
\centering
\includegraphics[width=0.9\linewidth]{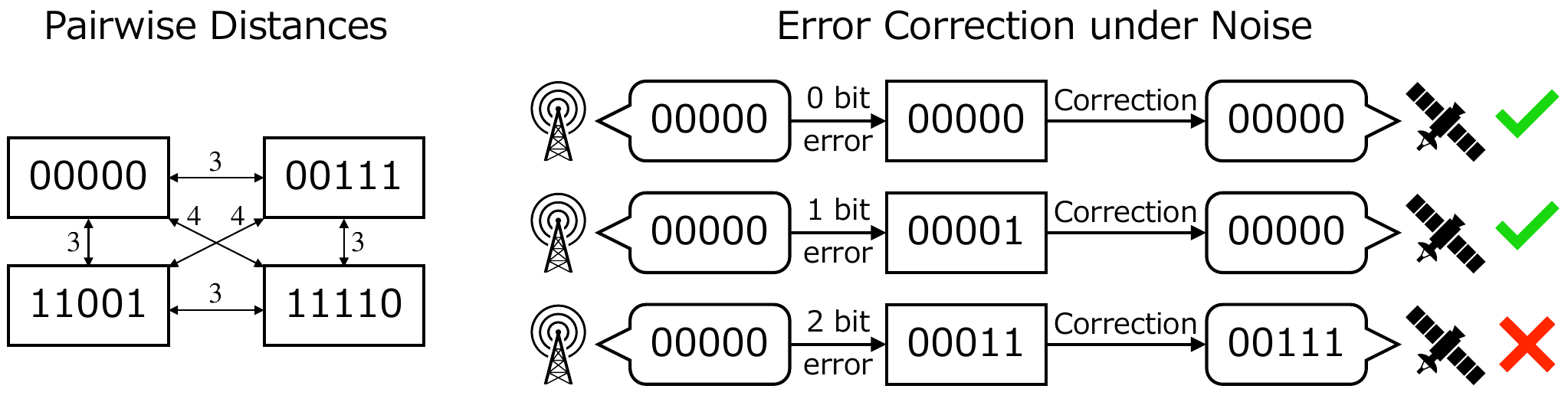}
\caption{Error correction as separation. Nearest-neighbor decoding corrects
any pattern of up to $\lfloor(d-1)/2\rfloor$ bit errors for a code with
minimum Hamming distance $d$. The code
$S=\{00000,00111,11001,11110\}$ has parameters $(n,M,d)=(5,4,3)$.
Left: Pairwise distances are at least three.
Right: Any single-bit error is corrected, whereas the illustrated two-bit
error causes incorrect decoding.}
\label{fig:error-correction-example}
\end{figure}

\begin{definition}[Binary Code Packing Problem]
\label{def:binary-code-packing}
Let $n,d$ be integers with $1\le d\le n$. The binary code packing problem,
whose optimal value is denoted by $A_2(n,d)$, is
\begin{equation*}
\underset{S\subseteq\{0,1\}^n}{\operatorname{maximize}}
\quad |S|
\quad\text{subject to}\quad d_H(x,y)\ge d
\quad\text{for all distinct }x,y\in S.
\end{equation*}
\end{definition}

\section{Our Results}

We provide three improved lower bounds on $A_2(n,d)$
in Proposition~\ref{prop:improved-lower-bounds}.

\begin{proposition}
\label{prop:improved-lower-bounds}
The following lower bounds hold:
\[
 A_2(22,9)\ge84,\qquad A_2(24,9)\ge196,\qquad A_2(25,11)\ge65.
\]
\end{proposition}

\begin{proof}
The complete sets in Section~\ref{sec:verification} contain, respectively, $84$, $196$, and
$65$ distinct words of the indicated lengths. Exhaustive comparison of
all unordered pairs gives minimum distances $9$, $9$, and $11$,
respectively. Each set is therefore a witness
for the corresponding bound.
\end{proof}

\begin{table}[H]
\centering
\caption{Comparison of lower bounds for binary codes.}
\label{tab:lower-bounds}
\begin{tabular}{ccr@{\ }lcc}
\toprule
$n$ & $d$ & \multicolumn{2}{c}{Previous lower bound} & Our results & Improvement\\
\midrule
$22$~bits & $9$~bits  & $80$~words  & \citep{ostergard}  & $84$~words  & $+4$~words\\
$24$~bits & $9$~bits  & $192$~words & \citep{elssel}     & $196$~words & $+4$~words\\
$25$~bits & $11$~bits & $64$~words  & \citep{levenshtein} & $65$~words  & $+1$~word\\
\bottomrule
\end{tabular}
\end{table}

The previous bounds in Table~\ref{tab:lower-bounds} are listed in
\citeauthor{brouwer}'s online table of general binary codes
(accessed October 1, 2026).
The bound for $(24,9)$ is obtained by shortening the construction of
\citet{elssel}.
For odd $d$, parity extension and puncturing give
$A_2(n,d)=A_2(n+1,d+1)$. Thus the same constructions give lower bounds
$84$, $196$, and $65$ for $(n,d)=(23,10),(25,10),(26,12)$.

At unchanged block length and minimum distance, the new codes support more
messages. For the 22-bit code, $k$ blocks support $(84/80)^k=1.05^k$ times
as many message sequences as the previous construction. For example,
100 blocks totaling 2,200 bits support approximately 131 times as many
message sequences, while guaranteeing correction of any four bit errors
per block.

\section{Verification and Explicit Codes}
\label{sec:verification}

Each code below is listed as whitespace-separated decimal integers. Each
integer represents an $n$-bit word, padded with leading zeros and written
with the most significant bit first.
All three codes pass the following checks, implemented using exact integer
arithmetic in Python 3.10 or later:
\begin{Verbatim}[baselinestretch=0.9,samepage=true,vspace=2pt,listparameters={\setlength{\partopsep}{0pt}}]
from itertools import combinations

def check(words, n, M, d):
    assert len(words) == len(set(words)) == M
    assert all(0 <= w < 2**n for w in words)
    distances = [(x ^ y).bit_count() for x, y in combinations(words, 2)]
    return min(distances) >= d
\end{Verbatim}

\begingroup
\setlength{\parskip}{2pt}
\makeatletter
\renewcommand\subsection{\@startsection{subsection}{2}{\z@}%
  {-6pt}{3pt}{\normalfont\large\bfseries}}
\makeatother
\fvset{vspace=2pt,listparameters={\setlength{\partopsep}{0pt}}}
\subsection*{$(22,84,9)$ code}
\begin{Verbatim}[fontsize=\normalsize,baselinestretch=0.858333333]
       0   178141   190058   204400   221583   297724   324867   408135   411064   497962
  601801   610614   665964   677523   755620   784475   814110   953681   974054   988173
 1010674  1086962  1108493  1117358  1137945  1306081  1312347  1341860  1419411  1431404
 1486646  1495241  1579221  1686456  1688647  1772291  1799420  1875855  1892464  1914402
 1927061  2064895  2134127  2158480  2177474  2189429  2330143  2365237  2471000  2476967
 2557155  2587420  2623931  2651716  2729843  2742412  2793686  2809641  2923026  2942445
 2949966  2970361  3138049  3171326  3310257  3330822  3349485  3367954  3481897  3497686
 3548812  3561843  3639364  3667899  3704092  3734243  3814823  3820120  3905226  3936224
 4110397  4123018  4133264  4157039
\end{Verbatim}

\subsection*{$(24,196,9)$ code}
\begin{Verbatim}[fontsize=\normalsize,baselinestretch=0.858333333]
       0    19421   118546   355397   363118   444713   570249   603974   621723   782447
  790750   849653  1029554  1160280  1242911  1326646  1340907  1437988  1458295  1584872
 1624259  1805188  1871295  1902960  1917613  2210677  2235929  2291762  2538830  2570625
 2585180  2672622  2843266  2852009  2916915  3018439  3032346  3129813  3224107  3310456
 3344311  3361898  3507523  3589636  3813617  3832620  3932131  3996120  4152500  4176543
 4303288  4331755  4430884  4446201  4529872  4597758  4611479  4893538  4940912  5005899
 5185804  5228327  5316838  5399457  5408138  5743122  5761487  5790845  5956316  5971770
 6027537  6198612  6217353  6249030  6369227  6378914  6450316  6596005  6614648  6646455
 6740964  6837072  6862203  7016471  7083564  7199038  7379859  7393358  7672317  7731158
 7754426  7874586  7974101  7989000  8154470  8162125  8228992  8334881  8410746  8556883
 8575630  8590748  8655783  8868064  8894155  9070855  9152021  9170376  9269052  9317499
 9357904  9448849  9662166  9668861  9875301  9888952  9907114 10047754 10097229 10348130
10377009 10459171 10474494 10494652 10536087 10707451 10774464 10856658 10871567 10953766
11099687 11122700 11148640 11313044 11326537 11344731 11557366 11638500 11656505 11803786
11860641 12017613 12100002 12118655 12133741 12227646 12410897 12460886 12632677 12651895
12869938 12879641 13091934 13122097 13136364 13221630 13242285 13468139 13499266 13547717
13655812 13678895 13714633 13892200 13958227 13972289 13990044 14230431 14271924 14321395
14418951 14438362 14726658 14881646 14939461 15105835 15183929 15321753 15370718 15397815
15564273 15649954 15664048 15681645 15862464 15876381 15961615 16025652 16135829 16199000
16239475 16465286 16479835 16643247 16667075 16693224
\end{Verbatim}

\begin{samepage}
\subsection*{$(25,65,11)$ code}
\begin{Verbatim}[fontsize=\normalsize,baselinestretch=0.858333333]
       0   129055   131040   932067  1592955  2867445  3299692  3475174  3957135  4037008
 4326367  5549386  5592405  5936553  5991862  6017602  6747174  6753337  6894810  7797465
 7944741  9872229 10788351 11154709 11250554 12448873 13390138 13479637 15163840 15255372
15845546 16648278 17671475 17996470 19387307 19442860 20865584 21905276 21940485 22111881
23351186 24355100 24942111 25323961 25325126 25400921 25974437 26556620 26935343 28325674
28332243 28430133 29233097 29358038 29463942 30319242 30852348 30882579 30997487 30997520
30999536 31758944 31810175 32201372 32202083
\end{Verbatim}
\end{samepage}

\endgroup
{
\makeatletter
\renewcommand\section{\@startsection{section}{1}{\z@}%
  {-8pt}{4pt}{\normalfont\Large\bfseries}}
\makeatother
\section*{Acknowledgements}
This work was partially supported by JST Moonshot R\&D Grant Number
JPMJPS2011 and by the New Energy and Industrial Technology Development
Organization (NEDO) under project JPNP26007. Kazuki Ota was supported by
JST SPRING, Grant Number JPMJSP2108. Takayuki Osa was supported by JSPS
KAKENHI Grant Number JP25K03176.

}

\end{document}